\documentclass[12pt]{amsart}
\usepackage{amssymb}
\usepackage{amsmath}
\usepackage{bbm}
\usepackage[dvips]{graphicx}
\usepackage{hyperref}
\usepackage[dvipsnames]{xcolor}
\usepackage{todonotes}

\numberwithin{equation}{section}

\newcommand{\1}{\hat{1}}
\newcommand{\nsetk}[2]{\begin{Bmatrix} #1 \\#2 \end{Bmatrix}}
\newcommand{\rpf}[4]{\nsetk{#1}{#2}_{#3,#4}}
\newcommand{\aver}[1]{\langle #1 \rangle}
\newcommand{\averr}[1]{\langle\!\langle #1 \rangle\!\rangle}

\newtheorem{thm}{Theorem}[section]
\newtheorem{prop}[thm]{Proposition}
\newtheorem{cor}[thm]{Corollary}
\newtheorem{lem}[thm]{Lemma}

\newtheorem{rem}[thm]{Remark}
\newtheorem{example}[thm]{Example}

\title[A simple symmetric exclusion process driven by a tracer]
{A simple symmetric exclusion process driven by an asymmetric tracer particle}
\author{Arvind Ayyer}
\address{Arvind Ayyer, Department of Mathematics, 
Indian Institute of Science, Bangalore 560012, India.}
\email{arvind@iisc.ac.in}

\subjclass[2010]{82C05, 82C22, 82C23, 60J27, 05A15, 05A18}
\keywords{exclusion process, tracer particle, exact solution, steady state, environment process, current, density, Stirling numbers}

\date{\today}

\begin{document}

\begin{abstract}
We consider an exclusion process on a periodic one-dimensional lattice where all particles perform simple symmetric exclusion at rate $1$ except for a single tracer particle, which performs partially simple asymmetric exclusion with rate $p$ to the right and rate $q$ to the left. This model was first considered by Ferrari, Goldstein and Lebowitz (Progr. Phys., 1985) as a test for the validity of the Einstein relation in microscopic systems.

The main thrust of this work is an exact solution for the steady state of this exclusion process. We show that the stationary probabilities factorize and give an exact formula for the nonequilibrium partition function. 
We provide formulas for the current and two-point correlations. 
When the tracer particle performs asymmetric exclusion ($q=0$), the results are shown to simplify significantly and we find an unexpected connection with the combinatorics of set partitions.
Finally, we study the system from the point of view of the tracer particle, the so-called environment process. In the environment process, we show that the density of particles decays with the scaled position in front of the tracer particle in the thermodynamic limit.
\end{abstract}

\maketitle

\section{Introduction}
\label{sec:intro}

We consider an exclusion process on a finite interval with periodic boundary conditions where all particles, except one, perform simple symmetric exclusion. 
The exceptional particle, which we call a {\em tracer particle} borrowing terminology from~\cite{burlatsky-etal-1992}, performs partially simple asymmetric exclusion with forward and backward rates $p$ and $q \leq p$, respectively. The other particles perform standard simple symmetric exclusion.

This process was one of several variants investigated by Ferrari, Goldstein and Lebowitz~\cite[Section IV(c)]{ferrari-goldstein-lebowitz-1985} to understand the validity of the Einstein relation in microscopic dynamics. Later, the dynamics of the tracer particle in such an exclusion process on $\mathbb{Z}$ was studied in a series of papers~\cite{burlatsky-etal-1992,burlatsky-etal-1996}. In particular, it was shown that the mean displacement of the tracer particle starting with Bernoulli initial conditions grows like $\sqrt{t}$, as opposed to $t$ in other models (where the motion is ballistic). 
A variant of this model, in which particles are also allowed to attach and escape from the lattice at some fixed rate, was also studied~\cite{benichou-etal-1999}, where the density of the other particles is calculated as seen from the point of view of the tracer particle.

Many rigorous results are known for the tracer particle when $p = q = 1$ on $\mathbb{Z}$. In that case, the tracer particle is indistinguishable from all the other particles and is commonly referred to as a `tagged particle' in the literature~\cite{ferrari-1986}. We use the term `tracer particle' also to avoid confusion.
Since the literature on tagged particles in exclusion process is immense and we are primarily interested in the case when $p \neq q$, we refer to Liggett's books~\cite[Chapter VIII.4]{liggett-ips-2005} and \cite[Chapter III.4]{liggett-sis-1999} for an introduction to this vast topic. We also note in passing that there have been studies on exclusion processes on $\mathbb{Z}$ where each particle carries a different rate~\cite{benjamini-ferrari-landim-1996,krug-ferrari-1996}, but the nature of the results is very different. 

The first rigorous result for this process on $\mathbb{Z}$ with $p \neq q$ is a law of large numbers for the displacement of the tracer particle~\cite{landim-olla-volchan-1998}, confirming the picture of ~\cite{burlatsky-etal-1992,burlatsky-etal-1996}. Later on, a central limit theorem for this displacement was also obtained~\cite{landim-volchan-2000}.
For a large class of processes on $\mathbb{Z}^d$, the Einstein relation has been proven to hold in \cite{komorowski-olla-2005}, but for a similar process as the one in this article on $\mathbb{Z}^d$, the Einstein relation has been disproven~\cite{loulakis-2005} for $d \geq 3$.

We obtain numerous exact results for the natural finite variant of this symmetric exclusion process with a single asymmetric tracer particle, which we now describe along with the plan of the rest of the paper. 
In Section~\ref{sec:def}, we define the model and state the main results. In Section~\ref{sec:ss}, we will prove the formula for the steady state probabilities. 
We will also derive formulas for the generating function of the nonequilibrium partition function there. The special case $q=0$ has a particularly illustrative combinatorial structure, and that will be dealt with in Section~\ref{sec:ss-q=0}. 
We next prove exact formulas for the current of particles as well as $1$-point and $2$-point correlations in Section~\ref{sec:corrfn}. 
The process, as seen from the point of view of the tracer particle, 
is known as the {\em environment process}. We will calculate the exact density profile both in front of and behind the tracer particle in Section~\ref{sec:profile}. Further, we will show that the density decays with the scaled position ahead of the tracer in the infinite volume limit, thus confirming predictions of~\cite{benichou-etal-1999,oshanin-etal-2004}. The calculation of the asymptotics used in Section~\ref{sec:profile} are relegated to Section~\ref{sec:asymp-pf}. The asymptotic formulas involve implicit functions for arbitrary $p$ and $q$, but they are expressible in terms of the Lambert W function when $p=1$ and $q=0$. This special case is discussed in Section~\ref{sec:asymp-special}. The infinite volume limit with finitely many particles is analyzed in Section~\ref{sec:asymp-finite}.

{\bf Note added in proof}: After this work was made public, Lobaskin and Evans uploaded a preprint~\cite{evans-lobaskin-2020} where they consider a similar model, but with many {\em totally asymmetric} tracer particles. We note however that they work in the totally asymmetric setting in which many of the computations simplify considerably. Unlike this work, they exploit the connection of this model to the zero range process to derive their results.

\section{Model definition and statement of results}
\label{sec:def}
Our model is an asymmetric exclusion process on $L$ sites with periodic boundary conditions, containing $n$ particles, one distinguished particle, which we call the {\em tracer particle} (a total of $n+1$ particles), and the remaining vacancies. 
The tracer particle will be denoted $\1$, the other particles will be denoted $1$ and vacancies will be denoted $0$. 

The dynamics is as follows. All particles except the tracer one perform simple symmetric exclusion with rate $1$, 
\begin{equation}
\label{dynamics-1}
10 \underset{1}{\overset{1}{\rightleftharpoons}} 01.
\end{equation}
The tracer particle performs simple asymmetric exclusion with rates as follows:
\begin{equation}
\label{dynamics-11}
\1 0 \underset{q}{\overset{p}{\rightleftharpoons}} 0 \1.
\end{equation}
We note that although there is no restriction on the values of $p$ and $q$ in principle, we will restrict ourselves to $0 \leq p,q \leq 1$ since the other particles cannot have velocities higher than 1 and the exclusion relation will force the tracer to also have a small velocity. Further, we will always fix $q \leq p$ without loss of generality.

Let $\Omega_{L,n}$ be the set of configurations of the ASEP with $L$ sites and $n$ particles. 
We note that our model is slightly more general than that of \cite{landim-olla-volchan-1998,landim-volchan-2000} because, in their language, we allow the exponential clock attached to the tracer particle to ring at rate $p+q$, which need not be equal to $1$.

\begin{example}
For example, with $L=3$ and $n = 1$, the column-stochastic generator for the process in the ordered basis
\[
\Omega_{3,1} = \{ (0, 1, \1), (0, \1, 1), (1, 0, \1), (1, \1, 0), (\1, 0, 1), (\1, 1, 0) \},
\]
is given by
\[
\left(
\begin{array}{cccccc}
 -p-1 & 0 & 1 & 0 & 0 & q \\
 0 & -q-1 & 0 & 1 & p & 0 \\
 1 & 0 & -q-1 & p & 0 & 0 \\
 0 & 1 & q & -p-1 & 0 & 0 \\
 0 & q & 0 & 0 & -p-1 & 1 \\
 p & 0 & 0 & 0 & 1 & -q-1 \\
\end{array}
\right).
\]
The steady state is then the null right-eigenvector and turns out to be
\[
\frac{1}{3 (p+q+2)} \left(q+1, p+1, p+1, q+1, q+1, p+1 \right).
\]
\end{example}

For any configuration $\tau = (\tau_1,\dots,\tau_L)$, let $\pi(\tau)$ denote the steady state probability of $\tau$. Since the local dynamics is independent of the position, we have the following translation invariance.

\begin{prop}
\label{prop:ss-trans-inv}
The steady state probabilities are invariant under translation, i.e.
\[
\pi(\tau_1, \tau_2, \dots, \tau_L) = \pi(\tau_2,\dots,\tau_L,\tau_1).
\]
\end{prop}

Furthermore, the steady state is reflection-invariant if $p=q$ by virtue of the dynamics of the particles in \eqref{dynamics-1} and \eqref{dynamics-11}. The more general statement is as follows.

\begin{prop}
\label{prop:ss-refl-inv}
The steady state probabilities are invariant under reflection and the interchange of $p$ and $q$, i.e.
\[
\pi(\tau_1, \tau_2, \dots, \tau_L) = \pi(\tau_L,\dots,\tau_2,\tau_1)\Big\vert_{p \leftrightarrow q}.
\]
\end{prop}

From Proposition~\ref{prop:ss-trans-inv}, it suffices to consider the stationary probabilities of configurations that begin with $\1$. For a configuration $\tau$ with $\tau_1 = \1$, define
\begin{equation}
\label{def-ss-wt}
w(\tau) = \prod_{\substack{i = 2 \\ \tau_i = 0}}^L
\left( 1 + p \, m_i(\tau) + q \, n_i (\tau) \right),
\end{equation}
where $m_i(\tau)$ (resp. $n_i(\tau)$) is the number of $1$'s
to the left (resp. right) of $i$ in $\tau$.

\begin{thm}
\label{thm:ss}
In the system with $L$ sites and $n$ $1$'s, the steady state probability of $\tau \in \Omega_{L,n}$ with $\tau_1 = \1$ is given by
\begin{equation}
\label{ss-formula}
\pi(\tau) = \frac{w(\tau)}{Z_{L,n}},
\end{equation}
where 
\[
Z_{L,n}(p,q) = \sum_{\tau \in \Omega_{L,n}} w(\tau)
\]
is the (nonequilibrium) partition function or normalization factor.
\end{thm}

The nature of the formula for $w(\tau)$ in \eqref{def-ss-wt} shows that the steady state in Theorem~\ref{thm:ss} is a nonequilibrium state if $p \neq q$, that is, there is no detailed balance. Equivalently, the process is irreversible.
Theorem~\ref{thm:ss} will be proved in Section~\ref{sec:ss}.

\begin{rem}
The factorization of the steady state in Theorem~\ref{thm:ss} can also be explained by a standard mapping to the zero-range process, for which the factorization is a well-established property~\cite{evans-hanney-2005}. However, we emphasize that deciphering the actual formula for the steady state is a nontrivial task, and this is made more complicated since we choose to work with arbitrary $p$ and $q$.
When $q=0$, the proofs simplify considerably.
\end{rem}

From Proposition~\ref{prop:ss-trans-inv} and Theorem~\ref{thm:ss}, it follows that $Z_{L,n}$ is $L$ times
a polynomial in $p$ and $q$ with integer coefficients. We thus define the {\em restricted partition function}, denoted by 
\begin{equation}
\label{pf-rpf-relation}
\rpf L{n+1}pq  = \sum_{\substack{\tau \in \Omega_{L,n} \\ \tau_1 = \1}} w(\tau) = \frac{Z_{L,n}(p,q)}L.
\end{equation}
To explain our notation for the restricted partition function, we recall some basic combinatorial facts.
Let $[n] := \{1,\dots,n\}$ and $\nsetk{[n]}{k}$ denote the collection of set partitions of $[n]$ into exactly $k$ parts. For example, 
\[
\nsetk{[4]}{2} = \{123|4, 124|3, 134|2, 1|234, 12|34, 13|24, 14|23 \}, 
\]
where we have divided the subsets in the set partition by vertical bars.
Set partitions of finite sets are well-studied combinatorial objects (see \cite[Section 6.1]{knuth-graham-patashnik-1994} for example), and the number of set partitions of $[n]$ into $k$ parts is known as the {\em Stirling number of the second kind}, denoted $\nsetk{n}{k}$. 
It is not difficult to see that this is a triangular sequence, i.e. $1 \leq k \leq n$, and it satisfies the recurrence relation
\begin{equation}
\label{stirling-recur}
\nsetk{n+1}k = \nsetk n{k-1} + k \nsetk nk,
\end{equation}
with $\nsetk n1 = \nsetk nn = 1$. One can check from \eqref{stirling-recur} that $\nsetk 42 = 7$ and this matches the above example. 

We analyze the special case of totally asymmetric motion of the tracer particle in Section~\ref{sec:ss-q=0}. 
We find that in the extreme case where the tracer particle performs totally asymmetric exclusion with rate $1$, the restricted partition function simplifies considerably. The following result, which explains our choice of notation, will be proved there. 

\begin{cor}
\label{cor:pf-p1q0}
The restricted partition function for the case $p=1$ and $q=0$ is given by the Stirling number of the second kind,
\[
\rpf L{n+1}10  = \nsetk {L}{n+1}.
\]
\end{cor}

To perform asymptotic analysis, it will be helpful to compute the generating function for the restricted partition function. As we will see in Corollary~\ref{cor:ss-p=q}, the stationary distribution is uniform when $p=q$. We will therefore restrict our attention to $p>q$
for this analysis.
The main result in this direction is an unexpectedly explicit formula for the bivariate generating function, which is exponential in the $L$ variable and ordinary in the $n$ variable.

\begin{thm}
\label{thm:rpf-egf}
The mixed bivariate generating function of the restrict\-ed partition function when $p > q$ is given by
\begin{equation}
\label{rpf-egf}
\sum_{L=1}^\infty \sum_{n=0}^{L-1} \rpf {L}{n+1}pq  x^n\frac{y^{L-1}}{(L-1)!}
= \exp \left( y + x \frac{\exp(p y) - \exp(q y)}{p-q} \right).
\end{equation}
\end{thm}

Theorem~\ref{thm:rpf-egf} will be proved in Section~\ref{sec:ss}.
Using Theorem~\ref{thm:rpf-egf}, we will show in Section~\ref{sec:asymp-pf} that the asymptotics of the 
restricted partition function $\rpf L{n+1}pq$ as $n, L \to \infty$ so that $n/L \to \rho \in (0,1)$ is given by 
\begin{multline}
\label{rpf-asymp}
\rpf L{\rho L + 1}pq \approx \frac{1}{\sqrt{2\pi L \left( \rho - (p \rho y_0 - 1) (q \rho y_0 - 1) \right)}}
\frac{\exp(y_0-1)}{y_0^{L-1}}  \\
 \times \left( \frac{L}{e} \right)^{L(1-\rho)-1}
\left( \frac{\exp(p y_0) - \exp(q y_0)}{\rho (p - q)} \right)^{\rho L},
\end{multline}
where $y_0$ is the unique positive real solution to the equation
\[
\exp((p-q)y) = \frac{\rho q y - 1}{\rho p y - 1},
\]
and where we use $a_L \approx b_L$ to mean $\lim_{L \to \infty} a_L/b_L = 1$. See Figure~\ref{fig:eg-pf-plot} for a comparison between the actual and asymptotic values of the restricted partition function. As $L$ gets larger, it is clearly seen that the ratio approaches 1.

\begin{center}
\begin{figure}[htbp!]
\includegraphics[scale=0.8]{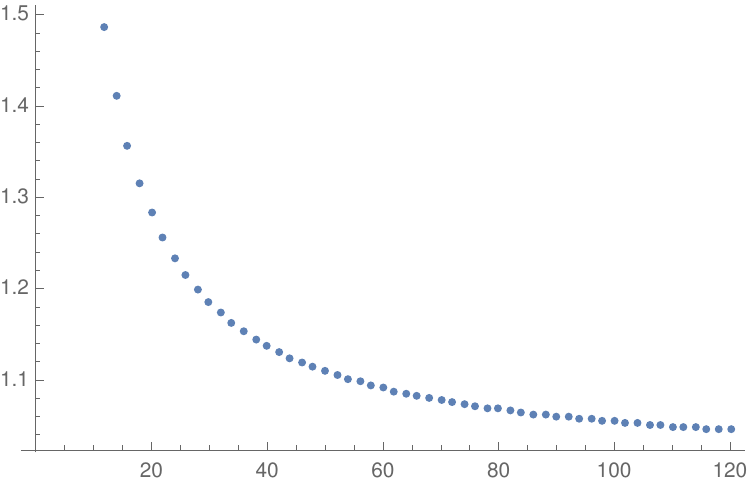}
\caption{A plot of the ratio of the asymptotic to the exact formula for the restricted partition function 
with $n = 0.5 L$, $p=0.55$ and $q=0.78$ for even values of $L$ ranging from 12 to 120. }
\label{fig:eg-pf-plot}
\end{figure}
\end{center}

From \eqref{rpf-asymp}, it is clear that the partition function in this case grows faster than exponentially in $L$. As a result, what is sometimes called the {\em nonequilibrium free energy}, defined by the limit
\[
\lim_{L \to \infty} \frac{\log Z_{L,\rho L}}{L}
\]
does not exist. Strictly speaking, the partition function is not well-defined because one could rescale all the $w(\tau)$'s arbitrarily.
However, if we insist that the GCD of all the weights $w(\tau)$ for $\tau \in \Omega_{L,n}$ is equal to $1$ to avoid spurious common factors, then we do not have any more freedom and the statement above is well-defined.

It is well understood that the nonequilibrium free energy is not a free energy in the sense of conventional equilibrium statistical mechanics. 
Therefore, its nonexistence does not violate any known laws.
However, in all exactly-solvable examples that we know of, this nonequilibrium free energy is well-defined for almost all points in the phase diagram. Moreover, the non-analyticity of this free energy signals phase transitions in the nonequilibrium steady state. This is the case, for example, for the open simple totally asymmetric exclusion process (TASEP)~\cite{dehp} and for many other models~\cite{be}.
In contrast, the free energy is not well-defined for any values of the parameters $p > q$ in our model. 

Using our results, we now compute important correlation functions in this process. We find an exact formula for the current. We denote the current of $1$'s and $\1$'s in the steady state by $J_1$ and $J_{\1}$, respectively.

\begin{thm}
\label{thm:current}
In the lattice with $L$ sites and $n$ 1's, the currents are given by
\[
J_{\1} = (p-q) \frac{\rpf {L-1}{n+1}pq }{Z_{L,n}}, \quad
J_1 = (p-q) \frac{n \, \rpf {L-1}{n+1}pq }{Z_{L,n}}.
\]
\end{thm}

Theorem~\ref{thm:current} will be proved in Section~\ref{sec:corrfn}.
We will also show in Section~\ref{sec:corrfn} that the asymptotic value of the current is zero when we take $L, n \to \infty$ so that $n/L \to \rho \in (0,1)$, i.e.,
\begin{equation}
\label{curr-asymp}
\lim_{L \to \infty} J_{\1} = \lim_{L \to \infty} J_1 = 0.
\end{equation}
This is not surprising since one does not expect that a single particle with a drift can generate a global current in a large system.

We are also interested in seeing the profile of particles from the point of view of the tracer particle. This is known as the {\em environment process}. 
By computing two-point correlation functions between the tracer particle and other particles, we obtain a formula for the density profile in the environment process. Let $\averr{\cdot}_{L,n}$ denote the expectation in the environment process. Since we will be interested in the density of particles both ahead of and behind the tracer particle, we will consider positions which are both positive and negative (relative to the tracer, which we will place at position $0$).

\begin{thm}
\label{thm:dens-from-test}
In the system with $L$ sites and $n$ $1$'s,
\begin{align*}
\averr{\tau_i}_{L,n} &= \sum_{j=0}^{L-n-1} \sum_{k=0}^j
\binom{L - 1 - i}k  \binom{i-1}{j-k} p^k q^{j-k} 
\frac{\rpf{L-j-1}{n}pq }{\rpf L{n+1}pq }, \\
\averr{\tau_{-i}}_{L,n} &= \sum_{j=0}^{L-n-1} \sum_{k=0}^j
\binom{L - 1 - i}k  \binom{i-1}{j-k} q^k p^{j-k} 
\frac{\rpf{L-j-1}{n}pq }{\rpf L{n+1}pq },
\end{align*}
for $1 \leq i \leq L-1$.
\end{thm}

Theorem~\ref{thm:dens-from-test} will be proved in Section~\ref{sec:profile} as a consequence of the two-point correlation functions between the tracer particle and other particles in Theorem~\ref{thm:dens-formula}. We now use \eqref{rpf-asymp} to calculate the density profile in the environment process in the limit of large system size and finite density $\rho$. 
Let $\averr \cdot$ be the distribution as seen from the tracer particle in the thermodynamic limit.
In Section~\ref{sec:profile}, we show that the density profile ahead of and behind the tracer particle at a distance $xL$, $x \in [0,1]$, is given by
\begin{equation}
\label{dens-asymp}
\begin{split}
\averr{\tau_{x L}} & \approx \rho y_0 (p - q) \frac{\exp(-(p-q) y_0 x)}{1 - \exp(-(p-q) y_0)}, \\
\averr{\tau_{-x L}} & \approx \rho y_0 (p - q) \frac{\exp((p-q) y_0 x)}{\exp((p-q) y_0) - 1}.
\end{split}
\end{equation}

As a test of our formula, we plot the exact density profile ahead of the tracer particle for a large system and the asymptotic formula in Figure~\ref{fig:eg-density-plot}, and we find very good agreement. This falloff in density has also been observed in a similar exclusion process with a tracer particle with adsorption and deposition~\cite[Fig. 2]{benichou-etal-1999}.

\begin{center}
\begin{figure}[htbp!]
\includegraphics[scale=0.4]{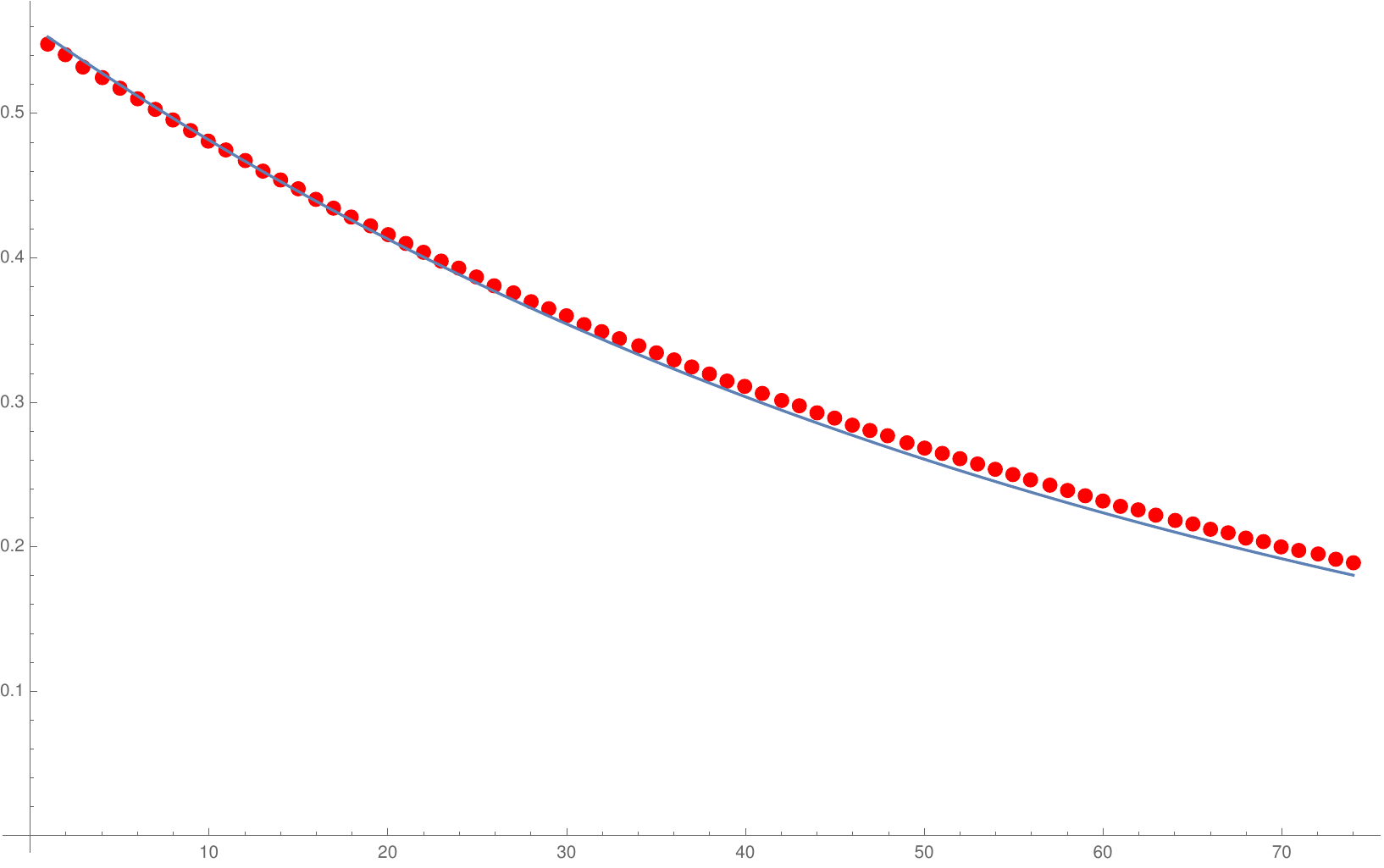}
\caption{A plot of the exact density of particles (red dots) ahead of the tracer particle in a system of size $L=75$ with $n=17$, $p=0.75$ and $q=0.4$, along with the expected curve from \eqref{dens-asymp} (blue curve). }
\label{fig:eg-density-plot}
\end{figure}
\end{center}

Detailed asymptotic studies of the partition function are performed in Section~\ref{sec:asymp-pf}. The analysis of \eqref{dens-asymp} for the interesting special case of $p=1,q=0$ is performed in Section~\ref{sec:asymp-special}

\section{Steady State}
\label{sec:ss}

We now compute the steady state for the model. Since the transition graph of this Markov process is identical to that of the single-species SSEP, it is clear that the process is ergodic if $q \neq 0$. For $q=0$, this can also be shown easily. Hence the steady state is unique.

\begin{proof}[Proof of Theorem~\ref{thm:ss}]
Since the steady state is unique, it suffices to verify \eqref{ss-formula} by the master equation,
\begin{equation}
\label{master-eq}
\sum_{\tau' \in \Omega_{L,n}} \text{rate}(\tau \to \tau') \pi(\tau)
= \sum_{\tau' \in \Omega_{L,n}} \text{rate}(\tau' \to \tau) \pi(\tau').
\end{equation}
The left-hand (resp. right-hand) side is the total outgoing (resp. incoming) contribution from (resp. to) $\tau$.

We write a generic state $\tau$ in block form as
\begin{equation}
\tau = \1 0^{\mu_1} 1^{\nu_1} \dots 0^{\mu_k} 1^{\nu_k} 0^{\mu_{k+1}},
\end{equation}
where $\mu_1, \mu_{k+1} \geq 0$ and all other $\mu_i, \nu_i > 0$.
There are four different kinds of states on whether $\mu_1$, $\mu_{k+1}$, both or none are zero.
We will only consider the case when $\mu_1 > 0$ and $\mu_{k+1} = 0$. The arguments for the other three cases are very similar and can be verified by the interested reader.

The total outgoing contribution in this case to \eqref{master-eq} is $(p+2k-1) \pi(\tau)$. Now let us consider the incoming contribution. For $1 \leq i \leq k-1$, there will be two contributions each involving the $1$'s in the boundary of the $i$'th block of $1$'s. There will be only one contribution for the $k$'th block, and one for the movement of $\1$. Therefore, we have $2k$ terms in the incoming contribution as well. We now analyze these.

Corresponding to the first block of $1$'s, we have
\begin{align*}
\tau^{(1)} &= \1 0^{\mu_1-1} \underset{\rightarrow}{1 0} 1^{\nu_1-1} \dots 0^{\mu_k} 1^{\nu_k} 0^{\mu_{k+1}}, \\
\tau^{(2)} &= \1 0^{\mu_1-1} 1^{\nu_1-1} \underset{\leftarrow}{0 1} 0^{\mu_2-1} \dots 0^{\mu_k} 1^{\nu_k} 0^{\mu_{k+1}},
\end{align*}
both of which make a transition to $\tau$ with rate $1$.
The arrows belows the sites indicate the transition that the $1$ has to make to reach $\tau$.
Applying \eqref{ss-formula}, we see that
\[
\frac{\pi(\tau^{(1)})}{\pi(\tau)} = \frac{1 + p + (n-1)q}{1 + n q}, 
\,
\frac{\pi(\tau^{(2)})}{\pi(\tau)} = \frac{1 + (\nu_1-1)p + (n-\nu_1+1)q}{1 + \nu_1 p + (n - \nu_1)q}.
\]
Considering the contribution to the second block of $1$'s similarly from configurations $\tau^{(3)}, \tau^{(4)}$, we find that
\begin{align*}
\frac{1 \cdot \pi(\tau^{(3)})}{\pi(\tau)} &= \frac{1 + (\nu_1 + 1)p + (n-\nu_1-1)q}{1 + \nu_1 p + (n - n-1)q}, \\
\frac{1 \cdot \pi(\tau^{(4)})}{\pi(\tau)} &= \frac{1 + (\nu_1 + \nu_2 -1)p + (n-\nu_1-\nu_2+1)q}{1 + (\nu_1 + \nu_2) p + (n - \nu_1 - \nu_2)q}.
\end{align*}
Now note that
\[
\frac{\pi(\tau^{(2)}) + \pi(\tau^{(3)})}{\pi(\tau)} = 2.
\]
Continuing this way, we find that the same equation will hold when we sum $\pi(\tau^{(2i)})$ and $\pi(\tau^{(2i+1)})$, for $2 \leq i <k$.
It only remains to analyze the transition from 
\[
\tau^{(2k)} = \underset{\leftarrow}{\;\1} 0^{\mu_1-1} 1^{\nu_1} \dots 0^{\mu_k} 1^{\nu_k} 0
\]
with rate $q$, where we have used Proposition~\ref{prop:ss-trans-inv} to ensure that the first site is $\1$.
Then, we have 
\[
\frac{q \cdot \pi(\tau^{(2k)})}{\pi(\tau)} = q \frac{1 + n p}{1 + n q}, 
\]
so that 
\[
\frac{\pi(\tau^{(1)}) + q \pi(\tau^{(2k)})}{\pi(\tau)} = 
\frac{1 + p + nq + n p q}{1 + n q} = 1 + p.
\]
In summary, the total incoming contribution is $(1+p + 2(k-1)) \pi(\tau)$, which is exactly the outgoing contribution, completing the proof in this case. The other three cases for $\mu_1$ and $\mu_{k+1}$ work in a completely analogous manner.
\end{proof}

By virtue of \eqref{def-ss-wt} and Proposition~\ref{prop:ss-refl-inv}, we have:

\begin{cor}
\label{cor:rpf-symmetric}
The restricted partition function $\rpf L{n+1}pq$ is a symmetric polynomial in $p$ and $q$ with integer coefficients.
\end{cor}

From the definition of the weight in \eqref{def-ss-wt}, the following special case immediately follows.

\begin{cor}
\label{cor:ss-p=q}
When $p = q$, the stationary distribution is uniform. Furthermore, it is easy to verify that the process is reversible in that case.
\end{cor}

Recall that $\nsetk{n}{k}$ is the Stirling number of the second kind.
For the restricted partition function with $n$ particles and a small number of vacancies, we can show by brute-force computation using \eqref{pf-rpf-relation} that
\begin{equation*}
\begin{split}
\rpf {n+2}{n+1}pq =&  n+1 + \nsetk {n+1}n (p+q) , \\
\rpf {n+3}{n+1}pq  =& \binom {n+2}2 + 3 \binom {n+2}3 (p+q)
+ \nsetk {n+2}n (p^2 + q^2) \\
&+ \frac{3n-1}{2} \binom{n+2}3 p q .
\end{split}
\end{equation*}
But the formulas for the restricted partition function get more complicated as $L-n$ gets larger. We now establish a recurrence that generalizes that of the Stirling numbers in \eqref{stirling-recur}.

\begin{prop}
\label{prop:pf-recur2}
The restricted partition function satisfies the recurrence relation
\begin{multline*}
\rpf L{n+1}pq = (1+n q) \rpf {L-1}{n+1}pq  
+ (1+p)^{L-n-1} \rpf {L-1}{n}{\frac p{1+p}}{\frac q{1+p}},
\end{multline*}
for $L > n \geq 0$, with the initial conditions $\rpf L 1 pq = \rpf L{L}pq = 1$.
\end{prop}

\begin{proof}
Decompose the set of words in $\Omega_{L,n}$ beginning with $\hat{1}$ into those beginning with $\hat{1}0$ and those with $\hat{1}1$. From the definition of $w$ in \eqref{def-ss-wt}, it is clear that the contribution of the first subset gives the first term on the right-hand side, the extra $0$ in front contributing $1+nq$.

The second subset requires a little more care. Consider a word $\tau = (\1, \tau_2, \dots, \tau_{L-1}) \in \Omega_{L-1,n-1}$, and form $\tau' \in \Omega_{L,n}$ as $\tau' = (\1, 1, \tau_2, \dots, \allowbreak \tau_{L-1})$. Comparing $w(\tau')$ and $w(\tau)$, we see that they have the same number of factors since the number of $0$'s in both is $L-n$. Moreover, each factor in $w(\tau')$ is $p$ plus the corresponding factor in $w(\tau)$. It is now easy to see that replacing $p$ by $p/(1+p)$ and $q$ by $q/(1+p)$ in $w(\tau)$ will give rise to the same factor as $w(\tau')$ except for an overall power of $1+p$. This proves the result.
\end{proof}

The column generating function for the Stirling numbers is given by\cite[Section 1.6]{wilf-1994} the product formula,
\begin{equation}
\label{stirling-columngf}
\sum_{n \geq k} \nsetk n k x^n = \frac{x^k}{(1-x)(1-2x) \cdots (1-kx)}.
\end{equation}
Now consider the generating function
\[
F_{n}(x) = \sum_{L = n+1}^\infty \rpf {L}{n+1}pq x^{L-n-1}
\]
in the formal variable $x$.
Then we have the following two-variable generalization of \eqref{stirling-columngf}.

\begin{thm}
\label{thm:rpf-ogf}
The generating function $F_{n}(x)$ is given by
\[
F_{n}(x) = \prod_{j=0}^{n} \frac{1}{1 - (1 + j p + (n - j) q) x}.
\]
\end{thm}

\begin{proof}
Clearly, the configurations are specified by the position of the $n$ $1$'s.
Using \eqref{def-ss-wt}, the desired generating function is
\begin{align*}
F_{n}(x) =& \sum_{m_0,\dots,m_n \geq 0} w( \hat{1} 0^{m_0} 1 0^{m_1} 
\dots 0^{m_{n-1}} 1 0^{m_n} ) x^n\\
=& \left( \sum_{m_0 \geq 0} (1 + n q)^{m_0} x^{m_0} \right)
\left( \sum_{m_1 \geq 0} (1 + p + (n-1) q)^{m_1} x^{m_1} \right) 
\dots \\
& \times \left(\sum_{m_n \geq 0} (1 + n p)^{m_n} x^{m_n} \right),
\end{align*}
which is easily summed to obtain the desired result.
\end{proof}

From the column generating function of the Stirling numbers in \eqref{stirling-columngf} one can derive the explicit formula (see 
\cite[Section 1.6]{wilf-1994})
\begin{equation}
\label{stirling-formula}
\nsetk n k = \sum_{r=1}^k (-1)^{k-r} \frac{r^n}{r!(k-r)!},
\end{equation}
which has the remarkable property of holding true even when $n < k$.
As an immediate corollary of Theorem~\ref{thm:rpf-ogf}, we find a similar expression for the restricted partition function.

\begin{cor}
\label{cor:rpf-formula}
The restricted partition function is given by
\begin{equation}
\label{rpf-formula}
\rpf {L}{n+1}pq = \sum_{j=0}^{n} \frac{(-1)^{n-j}}{j! (n-j)!}
\frac{(1 + j p + (n-j) q)^{L-1}}{(p-q)^{n}}.
\end{equation}
\end{cor}

\begin{proof}
Expand $F_{n}(x)$ in Theorem~\ref{thm:rpf-ogf} using partial fractions,
\begin{equation}
\label{pf-parfrac}
F_n(x) = \sum_{j=0}^n \frac{c_j}{1 - (1 + j p + (n - j) q) x}.
\end{equation}
Standard calculations show that 
\[
c_j = \frac{(-1)^{n-j}}{j! (n-j)!}
\left( \frac{1 + j p + (n - j) q}{p-q} \right)^n
\]
Expand the right hand side of \eqref{pf-parfrac} as a geometric series,
\[
F_n(x) = \sum_{m=0}^\infty \sum_{j=0}^n c_j (1 + j p + (n - j) q)^m x^m.
\]
The coefficient of $x^m$ is precisely $\rpf L{n+1}pq$ and this completes the proof. 
\end{proof}

We note that although Corollary~\ref{cor:rpf-formula} is explicit, it is not useful for asymptotic computations because the summands are not all positive. However, it is very useful for fast exact computations on a computer for numerical values of $p$ and $q$.

\begin{rem}
\label{rem:gf-property}
Corollary~\ref{cor:rpf-formula} can be seen as a two-variable generalization of \eqref{stirling-formula}. Moreover, 
the formula for the partition function in \eqref{rpf-formula} also gives the correct answer, namely $0$, when $L \leq n$ in complete analogy with \eqref{stirling-formula}. This is clear from the definition of the generating function $F_n(x)$, but it is not so obvious from \eqref{rpf-formula}.
\end{rem}

The mixed bivariate generating function of the Stirling numbers is given by the remarkably simple expression~\cite[Section 1.6]{wilf-1994},
\begin{equation}
\label{stirling-mixedgf}
\sum_{n \geq 0} \sum_{k=0}^n \nsetk n k \frac{x^n}{n!} y^k = \exp(y(\exp(x)-1)).
\end{equation}
Theorem~\ref{thm:rpf-egf} is a two-variable generalization of \eqref{stirling-mixedgf}, which we are now ready to prove.

\begin{proof}[Proof of Theorem~\ref{thm:rpf-egf}]
We first simplify the inner sum using Corollary~\ref{cor:rpf-formula},
\begin{align*}
\sum_{n = 0}^{L-1} \rpf {L}{n+1}pq  x^n =&
\sum_{n = 0}^{L-1} \sum_{j = 0}^n \frac{(-1)^{n-j}}{j! (n-j)!}
\frac{(1 + j p + (n - j) q)^{L-1}}{(p-q)^n} x^n \\
=& \sum_{j = 0}^{L-1} \sum_{n=j}^{L-1} \frac{(-1)^{n-j}}{j!(n-j)!}
\frac{(1 + j p + (n - j) q)^{L-1}}{(p-q)^n} x^n \\
=& \sum_{j = 0}^L \sum_{m=0}^{L-j} \frac{(-1)^m}{j! m!}
\frac{(1 + j p + m q)^{L-1}}{(p-q)^{j+m}}  x^{j+m},
\end{align*}
where we have interchanged the sums in the first step and shifted the inner sum in the second step. We now use Remark~\ref{rem:gf-property} to replace the index $L-1$ by $M$ for any $M$ larger than $L-1$. Since the equality holds for arbitrarily large values of $M$, we take the limit $M \to \infty$ to obtain
\[
\sum_{n = 0}^{L-1} \rpf {L}{n+1}pq  x^n =
\sum_{j = 0}^\infty \sum_{m=0}^{\infty} \frac{(-1)^m}{j! m!}
\frac{(1 + j p + m q)^{L-1}}{(p-q)^{j+m}}  x^{j+m}.
\]
We now plug this in \eqref{rpf-egf} to see that the left-hand side equals
\begin{align*}
& \sum_{L = 1}^\infty \sum_{j = 0}^\infty \sum_{m=0}^{\infty} 
\frac{(-1)^m}{j! m!} \frac{(1 + j p + m q)^{L-1}}{(p-q)^{j+m}}  x^{j+m} \frac{y^{L-1}}{(L-1)!} \\
&= \sum_{j = 0}^\infty \sum_{m=0}^{\infty} 
\frac{(-1)^m}{j! m!} \frac{x^{j+m}}{(p-q)^{j+m}} \exp(y(1 + j p + m q)) \\
&= \exp(y) \left( \sum_{j = 0}^\infty 
\frac{\exp(y p j) x^j}{j!(p-q)^j} \right)
\left( \sum_{m=0}^{\infty} 
\frac{(-1)^m \exp(y q m) x^m}{m!(p-q)^m} \right).
\end{align*}
These exponential sums are easily performed, leading to the desired result.
\end{proof}

\subsection{Special case: \texorpdfstring{$q=0$}{q=0}}
\label{sec:ss-q=0}

When the tracer particle moves totally asymmetrically ($q=0$), the 
restricted partition function $\rpf {L}{n+1}pq$ takes on a particularly combinatorial flavor. 
We generalize the Stirling numbers by setting
\[
\nsetk n{k}_p = \sum_{\pi \in \nsetk{[n]}{k}} p^{\#\{\text{number of elements in the subset containing 1 in $\pi$}\}-1}.
\]
Continuing the example in Section~\ref{sec:def}, $\nsetk 42_p = 1 + 3p + 3p^2$.
It can be easily verified that these polynomials satisfy a triangular recurrence relation (just like \eqref{stirling-recur} is satisfied by the Stirling numbers), 
\begin{equation}
\label{genstirling-recur}
\nsetk{n+1}k_p = \nsetk n{k-1}_p + (k-1+p) \nsetk nk_p,
\end{equation}
with initial conditions $\nsetk nn_p =1 $ and $\nsetk n1_p = p^{n-1}$.
Then we have the following result.

\begin{thm}
\label{thm:rpf-special-nsetk}
The restricted partition function in the case $q=0$ is given by
\[
\rpf {L}{n+1}p0 = p^{L-n-1} \nsetk {L}{n+1}_{1/p}.
\]
\end{thm}

\begin{proof}
We will verify that both sides satisfy the same recurrence relation and initial conditions. When $n = L-1$ and $n=0$, the right hand side is 1. 
Let us consider the left-hand sides. 
The only configuration in $\Omega_{L,L-1}$ is $\1 1 \dots 1$, and $\pi(\1 1 \dots 1) = 1$ by \eqref{ss-formula}.  Similarly, $\Omega_{L-n,0}$ has a single configuration with $\pi( \1 0 \dots 0) = 1$ by \eqref{ss-formula}. Thus, the left-hand side is also 1 in both cases.

Now, every configuration $\tau$ in $\Omega_{L,n}$ with $\tau_1 = \1$ is obtained by either appending $1$ to a configuration $\tau'$ in $\Omega_{L-1,n-1}$ with $\tau'_1 = \1$ or $0$ to a configuration $\tau''$ in $\Omega_{L-1,n}$ with $\tau''_1 = \1$.
In the former case, $\pi(\tau) = \pi(\tau')$ by \eqref{ss-formula} since $q=0$. In the latter case, $\pi(\tau) = (1 + n p) \pi(\tau'')$. We therefore obtain
\[
\rpf L{n+1}p0 = \rpf {L-1}{n}p0 + (1 + n p) \rpf {L-1}{n+1}p0,
\]
and complete the proof by noting that this is equivalent to \eqref{genstirling-recur}.
\end{proof}

Corollary~\ref{cor:pf-p1q0} now follows by setting $p=1$ in Theorem~\ref{thm:rpf-special-nsetk}. It can also be seen directly from \eqref{stirling-columngf} and Theorem~\ref{thm:rpf-ogf}

\section{Correlation functions in steady state}
\label{sec:corrfn}

The current of a particle across a bond is the amount per unit time it jumps across the bond in the forward direction minus that in the reverse direction. 
Let $\sigma$ (resp. $\tau$) denote the occupation variable for $\1$ (resp. $1$). That is to say, $\sigma_i = 1$ (resp. $\tau_i = 1$) if and only if the 
$i$'th site is occupied by a $\1$ (resp. $1$), and otherwise $\sigma_i$ (resp. $\tau_i$) is zero. 
We will denote averages with respect to the stationary distribution using angular brackets, $\aver{\cdot}_{L,n}$.

\begin{lem}
\label{lem:current-corr}
For $L > n \geq 0$ and $2 \leq i \leq L-1$, we have
\[
\aver{\sigma_1 \tau_i (1-\tau_{i+1})}_{L,n} - \aver{\sigma_1 (1-\tau_i) \tau_{i+1}}_{L,n}
= (p-q) \frac{\rpf {L-1}{n+1}pq}{Z_{L,n}(p,q)} \aver{\sigma_1 \tau_i}_{L-1,n}.
\]
\end{lem}

\begin{proof}
From the definitions,
\begin{multline*}
\aver{\sigma_1 \tau_k (1 - \tau_{k+1})}_{L,n}
- \aver{\sigma_1 (1 - \tau_{k}) \tau_{k+1} }_{L,n}  \\
= \frac{1}{ Z_{L,n}(p,q)}
\sum_{\alpha_1, \alpha_2} \left( w(\1 \alpha_1 10 \alpha_2) - w(\1 \alpha_1 01 \alpha_2) \right),
\end{multline*}
where the sum is over words $\alpha_1,\alpha_2$ in $\{0,1\}$ of lengths $(k-2)$ and $(L-1-k)$ respectively with a total of $(L-n-2)$ $0$'s and $(n-1)$ $1$'s.
The crucial observation is that
\[
w(\1 \alpha_1 10 \alpha_2) - w(\1 \alpha_1 01 \alpha_2) = (p-q) w(\1 \alpha_1 1 \alpha_2)
\]
for all possible words $\alpha_1,\alpha_2$, using \eqref{def-ss-wt}.
Now, each configuration on the right hand side is an element of $\Omega_{L-1,n}$ and conversely, every configuration in $\Omega_{L-1,n}$ occurs exactly once as $\alpha_1, \alpha_2$ vary.
Thus the sum is $\rpf {L-1}{n+1}pq \aver{\sigma_1 \tau_i}_{L-1,n}$, 
proving the result.
\end{proof}

We are now in a position to prove the formula for the current.
Recall that the current of $1$'s and $\1$'s in the steady state is denoted by $J_1$ and $J_{\1}$, respectively.

\begin{proof}[Proof of Theorem~\ref{thm:current}]
We first consider the current of $\1$ across a bond. Since the bond itself does not matter, we fix the bond to be the one between sites $1$ and $L$. Then the current of $\1$ is given by
\[
J_{\1} = p \aver{\sigma_{L} (1 - \tau_1 - \sigma_1)}_{L,n}
- q \aver{(1 - \tau_{L} - \sigma_{L}) \sigma_1}_{L,n}.
\]
By translation invariance of the stationary distribution (see Proposition~\ref{prop:ss-trans-inv}), we obtain
\[
J_{\1} = \frac{1}{Z_{L,n}(p,q)}
\sum_\alpha \left( p w(\1 0 \alpha) - q w(\1 \alpha 0) \right),
\]
where $w$ is the weight function defined in \eqref{def-ss-wt} and the sum is over all words $\alpha$ with $(L-n-2)$ $0$'s and $n$ $1$'s. Using the formula for the weight in \eqref{def-ss-wt}, this can be written as
\[
J_{\1} = \frac{1}{Z_{L,n}(p,q)}
\left( p (1+n q) \rpf {L-1}{n+1}pq - q (1 + n p) \rpf {L-1}{n+1}pq \right),
\]
which simplifies to give the result.

The current of $1$'s again across the bond between sites $L$ and $1$ is given by
\[
J_1 = \aver{\tau_{L} (1 - \tau_1 - \sigma_1)}_{L,n}
- \aver{(1 - \tau_{L} - \sigma_{L}) \tau_1}_{L,n}.
\]
By conditioning on the location of the tracer particle, we obtain
\begin{align*}
J_1 =& \sum_{k=2}^{L-1} \Big(
\aver{\tau_{L} (1 - \tau_1 - \sigma_1) \sigma_k}_{L,n}
- \aver{(1 - \tau_{L} - \sigma_{L}) \tau_1 \sigma_k}_{L,n} \Big), \\
=& \sum_{k=2}^{L-1} \Big(
\aver{\sigma_1 \tau_{L-k+1} (1 - \tau_{L-k+2})}_{L,n}
- \aver{\sigma_1 (1 - \tau_{L-k+1}) \tau_{L-k+2} }_{L,n} \Big),
\end{align*}
where we have used translation invariance again in the last step to ensure that the first site is occupied by $\1$. Now, replace $k$ by 
$L-k+1$ to obtain
\begin{align*}
J_1 &= \sum_{k=2}^{L-1} \Big(
\aver{\sigma_1 \tau_k (1 - \tau_{k+1})}_{L,n}
- \aver{\sigma_1 (1 - \tau_{k}) \tau_{k+1} }_{L,n} \Big), \\
&= (p-q) \frac{\rpf {L-1}{n+1}pq}{Z_{L,n}(p,q)} \sum_{k=2}^{L-1}  \aver{\sigma_1 \tau_k}_{L-1,n}
\end{align*}
by Lemma~\ref{lem:current-corr}. 
Now, the sum on the right hand side is just the total number of regular particles, which is $n$. This concludes the proof.
\end{proof}

Using the formula for the current in Theorem~\ref{thm:current}, we can compute the limiting value of the current of $\1$ and $1$'s. 
Let $L, n \to \infty$ so that $n/L \to \rho \in (0,1)$. We then use
\eqref{rpf-asymp} to compute the ratio of $\rpf {L-1}{n+1}pq$ and $\rpf {L}{n+1}pq$. It is reasonable to assume $\rho$ and $y_0$ (defined in \eqref{y0-def}) are the same for both restricted partition functions in the large size limit. We then find that
\[
J_{\1} \approx \frac{p-q}{L} \frac{y_0 e}{L} 
\left(1 - \frac{1}{L} \right)^{L(1 - \rho) -2}
\approx \frac{(p-q) y_0 e^\rho}{L^2},
\]
which tends to $0$ as $L \to \infty$. Using the same computation,
\[
J_1 \approx \frac{(p-q) y_0 \rho e^\rho}{L},
\]
which also tends to $0$. 

We now want to calculate other correlation functions.
Translation invariance of the stationary distribution (see Proposition~\ref{prop:ss-trans-inv}) gives the densities.

\begin{prop}
\label{prop:1pt-fn}
In the system with $L$ sites and $n$ $1$'s,
\[
\aver{\sigma_i}_{L,n} = \frac{1}{L}, \qquad
\aver{\tau_i}_{L,n} = \frac{n}{L}.
\]
\end{prop}

The computation of special two-point correlations, namely those involving the tracer particle, is more interesting. 
These will help understand the profile from the point of view of the tracer particle, i.e., in the environment process.

\begin{thm}
\label{thm:dens-formula}
In the system with $L$ sites and $n$ $1$'s,
\begin{equation}
\label{dens-formula}
\aver{\sigma_1 \tau_i}_{L,n} = 
\sum_{j=0}^{L-n-1} \sum_{k=0}^j
\binom{L - i}k  \binom{i-2}{j-k} p^k q^{j-k} 
\frac{\rpf{L-j-1}{n}pq }{Z_{L,n}(p,q) }, 
\end{equation}
for $2 \leq i \leq L$.
\end{thm}

\begin{proof}
We will prove this result by a double induction argument. One is an induction on $L$ and the other is a reverse induction on $i$. 
If $L = n+1$, then both $j$ and $k$ are forced to be $0$ in \eqref{dens-formula} and we obtain $1$, as expected. Now suppose $L > n+1$.
At the last site, we find that
\[
\aver{\sigma_1 \tau_{L}}_{L,n} = (1+q)^{L-n-1} \frac{\rpf {L-1}{n}{p/(1+q)}{q/(1+q)} }{Z_{L,n}(p,q)},
\]
by following the same steps as the proof of Proposition~\ref{prop:pf-recur2}.
Now, plug in the expression from Corollary~\ref{cor:rpf-formula} on the right hand side to obtain
\begin{equation}
\label{dens-last-lhs}
\aver{\sigma_1 \tau_{L}}_{L,n} =
\frac{(1+q)^{L-n-1}}{Z_{L,n}(p,q)}
\sum_{j=0}^{n-1} \frac{(-1)^{n-1-j}}{j!(n-1-j)!}
\frac{(1 + j p + (n-j)q)^{L-2}}{(p-q)^{n-1}}.
\end{equation}
We will now show that we obtain the same expression from \eqref{dens-formula}.
Substitute $i=L$ in the right hand side of \eqref{dens-formula} and again use Corollary~\ref{cor:rpf-formula} for 
$\rpf {L-k-1}{n}pq $ to get
\begin{multline*}
\sum_{k=0}^{L-n-1} \binom{L-2}k q^k \sum_{j=0}^{n-1}  \frac{(-1)^{n-1-j}}{j!(n-1-j)!}  
\frac{(1 + j p + (n-1-j)q)^{L-2-k}}{Z_{L,n}(p,q)(p-q)^{n-1}}.
\end{multline*}
Interchange the $j$ and $k$ sums to find the expression
\begin{multline*}
\sum_{j=0}^{n-1} \frac{(-1)^{n-1-j}}
{Z_{L,n}(p,q) j!(n-1-j)!(p-q)^{n-1}}   \\
 \times \sum_{k=0}^{L-n-1} \binom{L-2}k q^k (1 + j p + (n-1-j)q)^{L-2-k}.
\end{multline*}
Note that we can replace the upper limit of the $k$ sum by $L-2$ because $Z_{L,n} = 0$ for $L \leq n$ (see Remark~\ref{rem:gf-property}).
Then the inner sum becomes a standard binomial sum and we obtain the same expression as \eqref{dens-last-lhs}, proving the base case.

Now, we suppose that \eqref{dens-formula} is correct for $i+1$. To prove the formula at site $i$ we have to show that
\begin{multline}
\label{dens-to-prove}
\aver{\sigma_1 \tau_i}_{L,n} - \aver{\sigma_1 \tau_{i+1}}_{L,n} =
\sum_{j=0}^{L-n-1} \sum_{k=0}^j \Bigg( 
\binom{L - i}k  \binom{i-2}{j-k} \\
 - \binom{L-1 - i}k  \binom{i-1}{j-k} \Bigg) p^k q^{j-k} 
\frac{\rpf {L-j-1}{n}pq }{Z_{L,n}(p,q)  }.
\end{multline}
Now, we observe that
\[
\aver{\sigma_1 \tau_i}_{L,n} - \aver{\sigma_1 \tau_{i+1}}_{L,n}
= \aver{\sigma_1 \tau_i (1-\tau_{i+1})}_{L,n} - \aver{\sigma_1 (1-\tau_i) \tau_{i+1}}_{L,n}.
\]
Using Lemma~\ref{lem:current-corr},  we have
\[
\aver{\sigma_1 \tau_i}_{L,n} - \aver{\sigma_1 \tau_{i+1}}_{L,n}
= (p-q) \frac{\rpf {L-1}{n+1}pq}{Z_{L,n}(p,q)} \aver{\sigma_1 \tau_i}_{L-1,n}.
\]
By the induction assumption on $L$, use \eqref{dens-formula} for the right hand side to obtain
\begin{align*}
\aver{\sigma_1 \tau_i}_{L,n} &- \aver{\sigma_1 \tau_{i+1}}_{L,n}
= (p-q) \\
& \times \sum_{j=0}^{L-n-2} \sum_{k=0}^j 
\binom{L - 1 - i}k  \binom{i-2}{j-k} p^k q^{j-k} 
\frac{\rpf {L-j-2}{n}pq }{Z_{L,n}(p,q)}.
\end{align*}
First, shift $j$ to $j+1$ on the right hand side. We will now work purely with the $k$ sum.
Use $(p-q) p^k q^{j-k} = p^{k+1} q^{j-k} - p^k q^{j-k+1} $ to split the right hand side into two sums, and shift $k$ to $k-1$ in the first sum to get
\begin{align*}
\sum_{k=0}^j \Bigg( \binom{L -1 - i}{k-1}  \binom{i-2}{j-k}
- \binom{L - 1 - i}k  \binom{i-2}{j-k-1} \Bigg) p^k q^{j-k+1}.
\end{align*}
A little bit of algebra shows that the expression in the parentheses above can also be written as
\[
\binom{L - i}{k}  \binom{i-2}{j-k} - \binom{L - 1 - i}k  \binom{i-1}{j-k},
\]
leading to the summand in \eqref{dens-to-prove}. The limits of the $j$ sum are $1$ and $L-n-1$, but this does not cause a problem because the term $j=0$ in \eqref{dens-to-prove} contributes nothing. This proves the result.
\end{proof}

Theorem~\ref{thm:dens-formula} simplifies considerably when $p=1,q=0$. In that case, the summand $k$ is forced to be equal to $j$ and we obtain the following:

\begin{cor}
\label{cor:dens}
When $p=1$ and $q=0$, the density of particles in the system with $L$ sites and $n$ particles is given by
\[
\aver{\sigma_1  \tau_i}_{L,n} = \frac{1}{Z_{L,n}(1,0)} \sum_{j=0}^{L-n-1} \binom{L-i}{j} \nsetk{L-j-1}n,
\]
for $2 \leq i \leq L$.
\end{cor}

\section{Profile in the environment process}
\label{sec:profile}

In this section, we obtain the density profile as seen from the tracer particle, first for the finite system and then in the thermodynamic limit. 
We begin with the proof of Theorem~\ref{thm:dens-from-test} building on the results in the previous section.
Recall from Section~\ref{sec:def} that $\averr{\cdot}_{L,n}$ denotes the expectation in the environment process. 
Since all positions will be measured relative to the tracer particle, the position of the tracer particle is arbitrary and can be chosen to be $0$. Throughout this section, we will label the positions from $0$ to $L-1$.

\begin{proof}[Proof of Theorem~\ref{thm:dens-from-test}]
By the translation invariance in Proposition~\ref{prop:ss-trans-inv}, 
\[
\averr{\tau_i}_{L,n} = \sum_{j=0}^{L-1} \aver{\sigma_j \tau_{j+i}}
= L \aver{\sigma_0 \tau_{i}}.
\]
We now use the formula for the latter from Theorem~\ref{thm:dens-formula} (after replacing $i$ by $i+1$) to obtain $\averr{\tau_i}_{L,n}$ for $i > 0$. 

When looking at the densities behind the tracer, we will look at positions $-i$ for $1 \leq i \leq L-1$. Because of the circular geometry, position $-i$ is the same as $L-i$. 
We therefore again use the formula for $\aver{\sigma_0 \tau_{L-i+1}}$
from  Theorem~\ref{thm:dens-formula} to obtain $\averr{\tau_{-i}}_{L,n}$, completing the proof.
\end{proof}

Now, we use Theorem~\ref{thm:dens-from-test} to compute the density both in front of and behind the tracer particle in the infinite volume limit. 
We fix a position $i \approx xL$ in the system of size $L$, where $-1 \leq x \leq 1$, and there are $n \approx \rho L$ particles. 
We will compute the ratio of $\rpf {L-j-1}npq$ and $\rpf L{n+1}pq$ assuming that $j$ is fixed. This is a reasonable thing to do since the partition function grows superexponentially and therefore the ratio will die down faster than exponentially in $j$. Thus, one expects only a finite number of terms to contribute. When $j$ is finite, it is not difficult to see that the density of particles in the system with $L-j-1$ sites and $n-1$ particles is also $\rho$ to the lowest order. For the same reason, $y_0$ (defined in \eqref{y0-def}) is also the same to the lowest order,
Using \eqref{rpf-asymp}, we find that
\begin{equation}
\label{ratio-rpf}
\begin{split}
\frac{\rpf {L-j-1}npq}{\rpf L{n+1}pq} \approx &
 \frac{(L-j-1)^{L-j-2}}{(n-1)^{n-1}}
\frac{n^n}{L^{L-1}} \frac{\exp(j) y_0^{j+1} (p - q)}{\exp(p y_0) - \exp(q y_0)} \\
\approx & \frac{\rho y_0^{j+1} (p - q) L^{-j}}{\exp(p y_0) - \exp(q y_0)}.
\end{split}
\end{equation}
Recall that $\averr \cdot$ denotes the limiting distribution in the environment process.
Substituting \eqref{ratio-rpf} into the first formula in Theorem~\ref{thm:dens-from-test} with $0 \leq x \leq 1$, and approximating the sum by taking the $j$ limit to infinity, we find that the density in front of the particle is given by
\begin{equation}
\label{dens-asymp-front}
\begin{split}
\averr{\tau_{x L}} \approx & \sum_{j=0}^\infty \sum_{k=0}^j 
\frac{(x q)^{j-k} ((1-x)p)^k}{(j-k)! k!} \frac{\rho y_0^{j+1} (p - q)}{\exp(p y_0) - \exp(q y_0)} \\
= & \frac{\rho y_0 (p - q) }{\exp(p y_0) - \exp(q y_0)}
\sum_{j=0}^\infty \frac{y_0^j}{j!} 
\sum_{k=0}^j \binom{j}{k} (x q)^{j-k} ((1-x)p)^k \\
= & \frac{\rho y_0 (p - q) }{\exp(p y_0) - \exp(q y_0)}
\sum_{j=0}^\infty \frac{y_0^j}{j!} (x q + (1-x)p)^j \\
= & \rho y_0 (p - q) \frac{\exp(-(p-q) y_0 x)}{1 - \exp(-(p-q) y_0)},
\end{split}
\end{equation}
where we have used the binomial theorem to obtain the third line. Therefore, we have found that the density decays as a function of the scaled position ahead of the tracer particle. 
See Figure~\ref{fig:eg-density-plot} for a comparison with exact data for a small system.

To calculate the density behind the particle, we essentially repeat this calculation by plugging in \eqref{ratio-rpf} into the second formula in Theorem~\ref{thm:dens-from-test} with $0 \leq x \leq 1$. We then obtain
\begin{equation}
\label{dens-asymp-behind}
\averr{\tau_{-x L}} \approx \rho y_0 (p - q) \frac{\exp((p-q) y_0 x)}{\exp((p-q) y_0) - 1}.
\end{equation}
This validates the formulas in \eqref{dens-asymp}.

\section{Asymptotics of the restricted partition function}
\label{sec:asymp-pf}

We will first compute the asymptotics of $\rpf L{n+1}pq$ when $L,n$ become infinite with a finite density of particles. In Section~\ref{sec:asymp-special}, we will compute the asymptotics as in the previous section specialized to $p=1,q=0$. In Section~\ref{sec:asymp-finite}, we will fix $n$ and let $L$ approach infinity. The results in this section are not at the same level of rigor as in the other sections, but we believe they could be made rigorous with some effort. However, we have performed extensive numerical checks of the results to convince ourselves of their validity.

It is customary to take $L, n \to \infty$ so that there is a finite density of particles in the limit. Let $\rho = \lim_{n \to \infty} (n+1)/(L)$ be the limiting density of particles, which we assume to exist and satisfy $0 < \rho < 1$. We will first derive an asymptotic formula for the restricted partition function $\rpf L{n+1}pq$, closely following Temme's approach~\cite{temme-1993} for uniform asymptotics of the Stirling numbers.

From the generating function in Theorem~\ref{thm:rpf-egf}, it immediately follows that
\begin{equation}
\sum_{L \geq n+1} \rpf L{n+1}pq \frac{y^{L-1}}{(L-1)!} = \frac{\exp(y)}{n!}
\left( \frac{\exp(py) - \exp(qy)}{p - q} \right)^n.
\end{equation}
From this, we obtain the contour integral formula
\begin{equation}
\label{rpf-cont1}
\rpf L{n+1}pq = \frac{(L-1)!}{n!} \frac{1}{2 \pi i}
\oint_C \frac{\exp(y)}{y^{L}} \left( \frac{\exp(py) - \exp(qy)}{p-q} \right)^n \; \text{d}y,
\end{equation}
where $C$ is a small contour around the origin.
Write the integrand as $\exp(\phi(y))/y$, where using $n = \rho L$,
\begin{equation}
\label{phi-def}
\phi(y) = y + \rho L \log \left( \exp(py) - \exp(qy) \right) - 
\rho L \log (p-q) - (L-1) \log y.
\end{equation}
We now use the saddle point method to estimate the integral.
Setting $\phi' = 0$ and taking the large $L$ limit, we obtain the equation
\begin{equation}
\label{y0-def}
\exp((p-q)y) = \frac{\rho q y - 1}{\rho p y - 1}.
\end{equation}
Although $y=0$ is a solution, that is not of interest. 
It is an exercise in calculus to show that \eqref{y0-def} has a unique positive real solution, $y_0$, say.
The standard saddle point approximations do not give the best results, at least for $p=1,q=0$, and so we will follow Temme's strategy~\cite{temme-1993}.
For $y \to 0+$, $\exp(py) - \exp(qy) \sim (p-q)y$, and therefore
$\phi(y) \sim (n - L + 1) \log y$.
And, as $y \to \infty$, $\phi(y) \sim (1 + n p) y$, where we have assumed $p > q$.
These two limits suggest the transformation $y \to t(y)$ such that
\begin{equation}
\label{y to t}
\phi(y) = (1 + n p)  t + (n - L + 1) \log t + A,
\end{equation}
where $A$ does not depend on $t$. The derivative of the right hand side vanishes at $t_0 = \frac{L- n - 1}{1 + n p}$.
Then the following correspondences hold between the variables $y$ and $t$:
\begin{equation}
\begin{array}{|c||c|c|c|}
\hline
y & 0 & y_0 & \infty \\
\hline
t & 0 & t_0 & \infty \\
\hline
\end{array}.
\end{equation}
Substituting $y = y_0$ and $t = t_0$ into \eqref{y to t} shows that
\begin{equation}
\label{A-def}
A = \phi(y_0) - (1 + n p) t_0 + (L - n -1) \log t_0.
\end{equation}
Then a change of variables brings \eqref{rpf-cont1} to the form
\begin{equation}
\label{rpf-cont2}
\rpf L{n+1}pq = \frac{(L-1)!}{n!} \frac{\exp(A)}{2 \pi i}
\oint_{C'} \frac{\exp((1 + np) t) f(t)}{t^{L - n}} \; \text{d}t,
\end{equation}
where $f(t) = \frac{t}{y} \frac{\text{d}y}{\text{d}t}$ is analytic in a large domain of the complex plane including the origin. We deform the initial contour of a small circle around the origin to the contour $C'$ which passes through $t_0$. Differentiating \eqref{y to t} with respect to $t$ and using the definition of $t_0$ shows that
\begin{equation}
\label{f-def}
f(t) = \frac{(1 + n p)(t - t_0)}{y \phi'(y)}.
\end{equation}
We now apply the Cauchy integral formula to the contour integral \eqref{rpf-cont2} and obtain a first-order approximation to the restricted partition function
\begin{equation}
\label{rpf-approx}
\rpf L{n+1}pq \approx \binom{L-1}{n} \exp(A) f(t_0) (1 + n p)^{L-n}.
\end{equation}
To complete the analysis, we have to derive a formula for $\exp(A)$ and $f(t_0)$. The first is straightforward. From \eqref{A-def}, we get
\begin{equation}
\label{exp(A)-value}
\exp(A) = \left( \frac{\exp(py) - \exp(qy)}{p-q} \right)^n
\frac{\exp(y_0)}{y_0^{L-1}}
\left( \frac{L-n - 1}{e (1 + n p)} \right)^{L-n-1}.
\end{equation}
Applying L'H\^opital's rule to \eqref{f-def}, we see that
\begin{equation}
f(t_0) = \frac{1 + n p}{y_0 \phi''(y_0)} \left. \left( \frac{\text{d}y}{\text{d}t} \right|_{t = t_0, y = y_0} \right)^{-1}.
\end{equation}
But we know by definition of $f$ that
\[
\left. \frac{\text{d}y}{\text{d}t} \right|_{t = t_0, y = y_0} = 
\frac{f(t_0) y_0}{t_0},
\]
and substituting this above, we find that
\begin{equation}
\label{f(t_0)-computation}
f(t_0) = \frac{1}{y_0} \sqrt{\frac{(1 + n p) t_0}{\phi''(y_0)}}.
\end{equation}
Since we already know the value of $t_0$, all that remains is to calculate $\phi''(y_0)$. Taking the double derivative of \eqref{phi-def}, setting $y = y_0$ using \eqref{y0-def} and simplifying, we obtain
\begin{equation}
\begin{split}
\phi''(y_0) &= -\frac{\exp((p+q)y_0) (p-q)^2 (L-1-y)^2}{n y_0^2 \rho^2 (p \exp(p y_0) - q \exp(q y_0))^2} + \frac{L-1}{y_0^2} \\
& \approx \frac{1}{y_0^2} \left( L - 1 - \frac{(p \rho y_0 - 1) (q \rho y_0 - 1)}n
\right).
\end{split}
\end{equation} 
Plugging this into \eqref{f(t_0)-computation} and substituting the value of $t_0$, we get
\begin{equation}
\label{f(t_0)-value}
f(t_0) \approx \sqrt{\frac{\rho(1-\rho)}{\rho - (p \rho y_0 - 1) (q \rho y_0 - 1)}}.
\end{equation}
Finally, we substitute \eqref{f(t_0)-value} and \eqref{exp(A)-value} into 
\eqref{rpf-approx} and use the Stirling formula for the binomial coefficient to obtain the final result in \eqref{rpf-asymp}.

\subsection{Special case: \texorpdfstring{$p=1,q=0$}{p=1,q=0}}
\label{sec:asymp-special}

As before, we take the limit $L,n \to \infty$ so that $n/L \to \rho \in (0,1)$.
There are a lot of simplifications when $p=1$ and $q=0$, and it is worth going through this case in detail. By Corollary~\ref{cor:pf-p1q0}, the partition function is the Stirling number $\nsetk L{n+1}$.

The solution of \eqref{y0-def} turns out to be explicitly solved by $y_0 = 1/\rho - G$, where 
$G \equiv G(\rho) = - W_0 ( -\exp(-1/\rho)/\rho)$ and $W_0$ is the principal branch of the Lambert W function.
Recall that the Lambert W function is a family of functions defined by the inverse function of $f(z) = z \exp(z)$. It is multivalued since $f$ is not injective, which is why we consider the principal branch. 

By using known properties of the Lambert function, we can show that $\exp(y_0) = 1/(\rho G)$. Substituting $p=1, q=0$ and this value of $y_0$ in \eqref{rpf-asymp}, we obtain, after some manipulations,
\begin{equation}
\nsetk{L}{n + 1} \approx \frac{1}{\sqrt{2 \pi \rho L (1-G)}}
\left( \frac{\rho L}{1 - \rho G} \right)^{L-\rho L-1} 
\text{e}^{ \rho L(1 - G) + 1/\rho - G}.
\end{equation}
We note that we could obtain this result directly from~\cite{temme-1993},

The density at position $i \approx xL$ in front of the tracer particle is given by setting $p=1,q=0$ in \eqref{dens-asymp-front}. 
Again, using properties of the Lambert function, we find that the prefactor becomes 1, and therefore
\begin{equation}
\averr{ \tau_{xL} } \approx \exp(x(G - 1/\rho)) = (\rho G)^{x}.
\end{equation}
Thus, the density at position $xL$ decays with $x$ as $L \to \infty$. Moreover, for $x=0$, the density is 1. Suppose we take $i$ to be fixed, i.e. independent of $L$. Since the range of the sum in Corollary~\ref{cor:dens} is independent of $i$, we expect the answer to be the same as that for $i = 2$ when $L$ is large. Therefore, we find that the density at any fixed distance ahead of the tracer particle is precisely 1. 
This is intuitively clear since the tracer particle performs totally asymmetric motion and the exclusion interaction means that more and more particles accumulate as time grows. In the thermodynamic limit, this causes an infinite traffic jam! 
This justifies the heuristic picture proposed in \cite[Section 3, last paragraph]{oshanin-etal-2004}.
Similarly, the density at position $-i \approx -xL$ behind the tracer particle is given by setting $p=1,q=0$ in \eqref{dens-asymp-behind} and we obtain
\begin{equation}
\averr{ \tau_{-xL} } \approx (\rho G)^{1-x}.
\end{equation}
In contrast, the density immediately behind the tracer particle is not zero, but is given by $\rho G$. The function $\rho G$ grows monotonically with $\rho$ and is a concave function. See Figure~\ref{fig:rhoG-plot} to see the nature of the function.

\begin{center}
\begin{figure}[htbp!]
\includegraphics[scale=0.4]{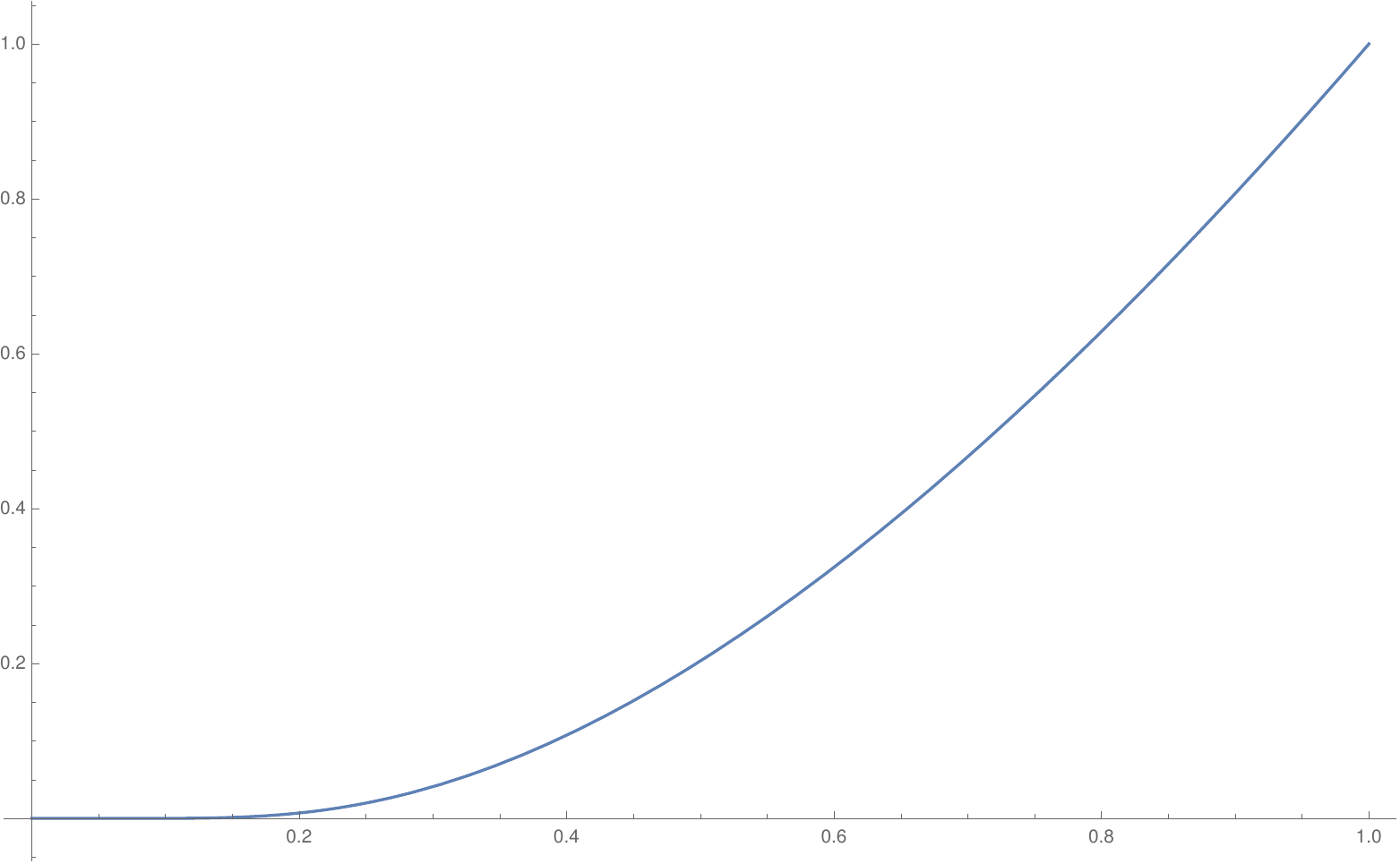}
\caption{A plot of $\rho G$ versus $\rho$ for small values of $\rho$. }
\label{fig:rhoG-plot}
\end{figure}
\end{center}

\subsection{Special case: finite number of particles}
\label{sec:asymp-finite}

We now consider the limit in which $n$ is fixed and we let the size of the system tend to infinity. 
We can assume $p \geq q$ as before. Since the case of $p = q$ is settled by Corollary~\ref{cor:ss-p=q}, we can assume $p > q$. In that case, we obtain the following asymptotic result.

From Theorem~\ref{thm:rpf-ogf}, the generating function $F_{n}(x)$ has poles of order one at $\{ 1/(1 + j p + (n - j)q)\}$ for $0 \leq j \leq n$. 
From the standard theory of analytic combinatorics (see \cite[Chapter VI]{flajolet-sedgewick-2009}, for example), it is well known that the pole closest to the origin contributes the most to the asymptotics. Since $p > q$ and all the poles lie on the positive real axis, the dominant contribution is from the smallest pole. This is given by $1/(1 + n p)$. Then the asymptotics are given by
\[
\rpf L{n+1}pq \approx \prod_{j=0}^{n-1} 
\frac{1}{1 - (1 + j p + (n - j)q )/ (1 + n p)}
\left( \frac{1}{1 + n p} \right)^{-L+n+1},
\]
which simplifies to 
\[
\rpf L{n+1}pq \approx \frac{(1 + n p)^{L -2}}{n! (p-q)^{n-1}} \quad \text{as $L \to \infty$.}
\]
Using this expression and Theorems~\ref{thm:current} and \ref{thm:dens-formula}, one can easily verify that the current and the density go to zero as $L \to \infty$, as it was expected.

\section*{Acknowledgements}
The author was partially supported by Department of Science and Technology grant EMR/2016/006624 and by the UGC Centre for Advanced Studies. We thank Anupam Kundu, Joel Lebowitz and E. R. Speer for helpful discussions.

\bibliographystyle{alpha}
\bibliography{ssep,asep}

\newcommand{\etalchar}[1]{$^{#1}$}
\begin{thebibliography}{BOMM92}

\bibitem[BCL{\etalchar{+}}99]{benichou-etal-1999}
O.~B{\'e}nichou, A.~M. Cazabat, A.~Lemarchand, M.~Moreau, and G.~Oshanin.
\newblock Biased diffusion in a one-dimensional adsorbed monolayer.
\newblock {\em Journal of Statistical Physics}, 97(1):351--371, Oct 1999.

\bibitem[BE07]{be}
R.~A. Blythe and M.~R. Evans.
\newblock Nonequilibrium steady states of matrix-product form: a solver's
  guide.
\newblock {\em J. Phys. A}, 40(46):R333--R441, 2007.

\bibitem[BFL96]{benjamini-ferrari-landim-1996}
I.~Benjamini, P.~A. Ferrari, and C.~Landim.
\newblock Asymmetric conservative processes with random rates.
\newblock {\em Stochastic Process. Appl.}, 61(2):181--204, 1996.

\bibitem[BOMM92]{burlatsky-etal-1992}
S.F. Burlatsky, G.S. Oshanin, A.V. Mogutov, and M.~Moreau.
\newblock Directed walk in a one-dimensional lattice gas.
\newblock {\em Physics Letters A}, 166(3):230 -- 234, 1992.

\bibitem[BOMR96]{burlatsky-etal-1996}
S.~F. Burlatsky, G.~Oshanin, M.~Moreau, and W.~P. Reinhardt.
\newblock Motion of a driven tracer particle in a one-dimensional symmetric
  lattice gas.
\newblock {\em Phys. Rev. E}, 54:3165--3172, Oct 1996.

\bibitem[DEHP93]{dehp}
B.~Derrida, M.~R. Evans, V.~Hakim, and V.~Pasquier.
\newblock Exact solution of a {$1$}{D} asymmetric exclusion model using a
  matrix formulation.
\newblock {\em J. Phys. A}, 26(7):1493--1517, 1993.

\bibitem[EH05]{evans-hanney-2005}
M~R Evans and T~Hanney.
\newblock Nonequilibrium statistical mechanics of the zero-range process and
  related models.
\newblock {\em Journal of Physics A: Mathematical and General},
  38(19):R195--R240, apr 2005.

\bibitem[Fer86]{ferrari-1986}
Pablo~A. Ferrari.
\newblock The simple exclusion process as seen from a tagged particle.
\newblock {\em Ann. Probab.}, 14(4):1277--1290, 1986.

\bibitem[FGL85]{ferrari-goldstein-lebowitz-1985}
Pablo~A. Ferrari, Sheldon Goldstein, and Joel~L. Lebowitz.
\newblock Diffusion, mobility and the {E}instein relation.
\newblock In {\em Statistical physics and dynamical systems ({K}\"{o}szeg,
  1984)}, volume~10 of {\em Progr. Phys.}, pages 405--441. Birkh\"{a}user
  Boston, Boston, MA, 1985.

\bibitem[FS09]{flajolet-sedgewick-2009}
Philippe Flajolet and Robert Sedgewick.
\newblock {\em Analytic combinatorics}.
\newblock Cambridge University Press, Cambridge, 2009.

\bibitem[GKP94]{knuth-graham-patashnik-1994}
Ronald~L. Graham, Donald~E. Knuth, and Oren Patashnik.
\newblock {\em Concrete mathematics}.
\newblock Addison-Wesley Publishing Company, Reading, MA, second edition, 1994.
\newblock A foundation for computer science.

\bibitem[KF96]{krug-ferrari-1996}
Joachim Krug and Pablo~A Ferrari.
\newblock Phase transitions in driven diffusive systems with random rates.
\newblock {\em Journal of Physics A: Mathematical and General},
  29(18):L465--L471, sep 1996.

\bibitem[KO05]{komorowski-olla-2005}
Tomasz Komorowski and Stefano Olla.
\newblock On mobility and {E}instein relation for tracers in time-mixing random
  environments.
\newblock {\em J. Stat. Phys.}, 118(3-4):407--435, 2005.

\bibitem[LE20]{evans-lobaskin-2020}
Ivan Lobaskin and Martin~R Evans.
\newblock Driven tracers in a one-dimensional periodic hard-core lattice gas.
\newblock {\em arXiv preprint arXiv:2002.02841}, 2020.

\bibitem[Lig99]{liggett-sis-1999}
Thomas~M. Liggett.
\newblock {\em Stochastic interacting systems: contact, voter and exclusion
  processes}, volume 324 of {\em Grundlehren der Mathematischen Wissenschaften
  [Fundamental Principles of Mathematical Sciences]}.
\newblock Springer-Verlag, Berlin, 1999.

\bibitem[Lig05]{liggett-ips-2005}
Thomas~M. Liggett.
\newblock {\em Interacting particle systems}.
\newblock Classics in Mathematics. Springer-Verlag, Berlin, 2005.
\newblock Reprint of the 1985 original.

\bibitem[Lou05]{loulakis-2005}
Michail Loulakis.
\newblock Mobility and {E}instein relation for a tagged particle in asymmetric
  mean zero random walk with simple exclusion.
\newblock {\em Ann. Inst. H. Poincar\'{e} Probab. Statist.}, 41(2):237--254,
  2005.

\bibitem[LOV98]{landim-olla-volchan-1998}
C.~Landim, S.~Olla, and S.~B. Volchan.
\newblock Driven tracer particle in one-dimensional symmetric simple exclusion.
\newblock {\em Comm. Math. Phys.}, 192(2):287--307, 1998.

\bibitem[LV00]{landim-volchan-2000}
Cl\'{a}udio Landim and S\'{e}rgio~B. Volchan.
\newblock Equilibrium fluctuations for a driven tracer particle dynamics.
\newblock {\em Stochastic Process. Appl.}, 85(1):139--158, 2000.

\bibitem[OBBM04]{oshanin-etal-2004}
G.~Oshanin, O.~B\'{e}nichou, S.~F. Burlatsky, and M.~Moreau.
\newblock Biased tracer diffusion in hard-core lattice gases: some notes on the
  validity of the {E}instein relation.
\newblock In {\em Instabilities and nonequilibrium structures {IX}}, volume~9
  of {\em Nonlinear Phenom. Complex Systems}, pages 33--74. Kluwer Acad. Publ.,
  Dordrecht, 2004.

\bibitem[Tem93]{temme-1993}
N.~M. Temme.
\newblock Asymptotic estimates of {S}tirling numbers.
\newblock {\em Stud. Appl. Math.}, 89(3):233--243, 1993.

\bibitem[Wil94]{wilf-1994}
Herbert~S. Wilf.
\newblock {\em generatingfunctionology}.
\newblock Academic Press, Inc., Boston, MA, second edition, 1994.

\end{thebibliography}

\end{document}